\documentclass[11pt,twoside]{article}
\usepackage{acta-info}
\usepackage{euler}

\newcommand{\vol}{\textbf{4,} 1 (2012) 33--47}   
\newcommand{\amss}{\amssubj{
68R05
}}
\newcommand{\CRs}{\CRsubj{
G.2.1
}}
\newcommand{\keyw}{\keywords{
sumset, characteristic function, partition, glanular set
}}

\begin{document}

\title{
Counting $(k,l)$-sumsets in groups of prime order
}

\maketitle

\oneauthor{
Vahe SARGSYAN
}{
\href{http://www.msu.ru}{Moscow State University}
}{
\href{mailto:vahe\_sargsyan@ymail.com}{vahe\_sargsyan@ymail.com}
}

\short{
V. Sargsyan
}{
Counting $(k,l)$-sumsets in groups of prime order
}

\begin{abstract}
A subset $A$ of a group $\textbf{G} $ is called $(k, l)$-{\it sumset}, if $A= kB-lB$ for some $B\subseteq \textbf{G}$, where $kB-lB=\{x_1+\dots+x_k-x_{k+1}-\dots-x_{k+l} : x_1,\dots, x_{k+l}\in B\}.$
Upper and lower bounds  for the number $(k, l)$-sumsets in groups of  prime order are provided.
\end{abstract}

\noindent

\section{Introduction}
Let $p$ be a prime number and $k$, $l$ be nonnegative integers with $k + l \geq 2 $. Write $\textbf{Z}_p$ for the group of residues modulo $p$. A subset $A\subseteq \textbf{Z}_p$ is called $(k, l)$-{\it sumset}, if $A=kB-lB$ for some $B\subseteq \textbf{Z}_p$, where $kB-lB=\{x_1+\dots+x_k-x_{k+1}-\dots-x_{k+l} : x_1,\dots, x_{k+l}\in B\}.$
Write $ \textbf{SS}_{k, l}(\textbf{Z}_p)$ for the collection of $(k, l)$-sumsets in $\textbf{Z}_p.$
\par B. Green and I. Ruzsa in ~\cite{GRUZ1} proved
\[
p^2 2^{p/3}\ll |\textbf{SS}_{2, 0}(\textbf{Z}_p)|\leq 2^{p/3 +\theta(p)}
\]
where $\theta(p)/p \rightarrow 0$ as $p \rightarrow \infty$ and $\theta(p)\ll p{(\log_{}{\log_{}p})}^{2/3}{(\log_{}{p})}^{-1/9}$ (hereafter logarithms are to base two).
\par The aim of this work is to obtain bounds for the number $|\textbf{SS}_{k, l}(\textbf{Z}_p)|.$ We prove
\begin{theorem}\label{t6}
Let $p$ be a prime number and $k$,$l$ be nonnegative integers with $k + l \geq 2$. Then there exists a positive constant $\textbf{C}_{k, l}$ such that
\begin{equation}\label{glux}
\textbf{C}_{k,l}2^{p/(2(k+l)-1)} \leq |\textbf{SS}_{k, l}(\textbf{Z}_p)| \leq 2^{(p/(k+l+1)) + (k+l-2) + o(p)}.
\end{equation}
\end{theorem}

\section{Definitions and auxiliary results}
Let $ \textbf{R}$ be the set of real numbers, $f_i: \textbf{Z}_p \to \textbf{R}$, $i=1, \dots , m,$ and $x\in \textbf{Z}_p$. We set
\[
(f_1 * \dots * f_m)(x)=
\]
\begin{equation}\label{7892}
=\sum_{x_1\in \textbf{Z}_p} \dots \sum_{x_{m-1}\in \textbf{Z}_p}{f_1(x_1) \dots  f_{m-1}(x_{m-1}) f_m(x- x_1 - \dots - x_{m-1})}
\end{equation}
and
\[
\widehat{f}(x) = \sum_{y\in \textbf{Z}_p}{f(y) e^{2\pi i \frac {xy}{p}}}.
\]
The function $\widehat{f}(x)$ is called {\it Fourier transform} of $f.$
\begin{lemma}\label{09876} We have
\begin{equation}\label{01987199}
(\widehat{f_1 * \dots * f_m })(x) =\widehat{f_1}(x) \dots  \widehat{f_m}(x).
\end{equation}
\end{lemma}
\begin{proof}
By definition
\[
(\widehat{f_1 * \dots * f_m })(x) = \sum_{y\in \textbf{Z}_p}{(f_1 * \dots * f_m)(y)e^{2\pi i\frac{yx}{p}}} =\]
\[ = \sum_{y\in \textbf{Z}_p}\sum_{y_1\in \textbf{Z}_p} \dots \sum_{y_{m-1}\in \textbf{Z}_p}{f_1(y_1) \dots f_{m-1}(y_{m-1})}\times
\]
\[
\times f_{m}(y - y_1 - \dots - y_{m-1})\cdot e^{2\pi i\frac{y_{1}x}{p}} \dots e^{2\pi i\frac{y_{m-1}x}{p}}\cdot e^{2\pi i\frac{(y-y_1- \dots -y_{m-1})x}{p}} =
\]
\[
= \sum_{y_1\in \textbf{Z}_p}{f_1(y_1)}\cdot e^{2\pi i\frac{y_{1}x}{p}} \dots \sum_{y_{m-1}\in \textbf{Z}_p}{f_{m-1}(y_{m-1})}\cdot e^{2\pi i\frac{y_{m-1}x}{p}}\times
\]
\[
\times \sum_{y\in \textbf{Z}_p}{f_{m}(y-y_1- \dots -y_{m-1})}\cdot e^{2\pi i\frac{(y-y_1- \dots -y_{m-1})x}{p}}=
 \widehat{f_1}(x) \dots \widehat{f_m }(x).
\]
\end{proof}
Denote the characteristic function of a set $A$ by $ \chi_{A}(x).$ Let $ A_1, \dots, A_m $ be non-empty subsets of $ \textbf{Z}_p$. Then $(\chi_{A_1} \ast \dots \ast \chi_{A_m})(x)$ will be the number of vectors  $(x_1, \dots, x_m) \in A_1 \times \dots \times A_m$ such that $x \equiv x_1 + \dots + x_m \pmod{p}$. Set $ A_1 + \dots + A_m =\{x_1 + \dots + x_m \pmod{p}: x_1 \in A_1, \dots, x_m \in A_m \}.$ We define $ S_{h, m}(A_1, \dots, A_m) = \{x \in \textbf{Z}_p: (\chi_{A_1} \ast \dots \ast \chi_{A_m})(x) \geq h\}, $ where $h> 0$. Further, for any integer $i$ and any $A \subseteq \textbf{Z}_p $ denote the set $\underbrace{A + \dots + A}_i$ by $iA$, and the set $\{p-x: x \in A \}$ by $-A $.
\begin{theorem}\label{t7}
\emph{(Cauchy-Davenport, \cite{NATH})}. Let $ A_1, \dots, A_m $ be non-empty subsets of $ \textbf{Z}_p$. Then $|A_1 + \dots + A_m|\geq \min(p,|A_1| + \dots + |A_m| - (m-1)).$
\end{theorem}
\begin{theorem}\label{t8}
$(Pollard, \cite{POLLARD})$. Let $ A_1, A_2 $ be non-empty subsets of $ \textbf{Z}_p$. Then
\[
|S_{1, 2}(A_1,A_2)| + \dots  + |S_{t, 2}(A_1,A_2)| \geq t \min(p,|A_1| + |A_2| -t),
\]
where  $t\leq \min(|A_1|,|A_2| ).$
\end{theorem}
Theorems \ref{t7}, \ref{t8} imply the following two statements.
\begin{lemma}\label{t9}
Let $ A_1, \dots, A_m $ non-empty subsets of $ \textbf{Z}_p$. Then
\[
|S_{1, m}(A_1, \dots , A_m)| + \dots + |S_{t, m}(A_1, \dots , A_m)| \geq \] \[\geq t \min(p,|A_1| + \dots + |A_m| - t - m + 2),
\]
where $t\leq \min(|A_1|, \dots , |A_m|).$
\end{lemma}
\begin{proof}
Without loss of generality we assume $|A_1|= \min(|A_1|, \dots , |A_m|)$. By Theorem \ref{t8} we have
\[
|S_{1, 2}(A_1, (A_2 + \dots + A_m))| + \dots + |S_{t, 2}(A_1,( A_2 + \dots + A_m))| \geq\]
\begin{equation}\label{1}\geq t\min(p,|A_1| + |A_2 + \dots + A_m| -t),
\end{equation}
where $t\leq |A_1|.$\\
On the other hand by Theorem \ref{t7} we have
\begin{equation}\label{2} |A_2 + \dots + A_m|\geq \min(p,|A_2| + \dots + |A_m| - (m-2)). \end{equation}
Substituting (\ref{2}) in (\ref{1}), we obtain
\[
|S_{1, m}(A_1, \dots , A_m)| + \dots + |S_{t, m}(A_1, \dots , A_m)| \geq\] \[\geq |S_{1, 2}(A_1, (A_2 + \dots + A_m))| + \dots + |S_{t, 2}(A_1,( A_2 + ... + A_m))| \geq \]
\[
\geq t \min(p,|A_1| + \dots + |A_m| - t - m + 2).
\]
\end{proof}
\begin{lemma}\label{t10}
Let $ A_1, \dots, A_m $ be non-empty subsets of $ \textbf{Z}_p$ and $h\leq \min{(|A_1|, \dots ,|A_m|)}$. Then
\[
|S_{h, m}(A_1, \dots , A_m)|\geq \min(p,|A_1| + \dots + |A_m| - m + 2) - 2(hp)^{1/2}.
\]
\end{lemma}
\begin{proof}
Note that $|S_{i, m}(A_1, \dots , A_m)| \geq |S_{j, m}(A_1, \dots , A_m)|$ for $i\leq j.$ Choose $h \leq t \leq \min{(|A_1|, \dots ,|A_m|)}$. By Lemma \ref{t9} we have
\[
t \min(p,|A_1| + \dots + |A_m| - t - m + 2) \leq \] \[ \leq |S_{1, m}(A_1, \dots , A_m)| + \dots + |S_{t, m}(A_1, \dots , A_m)| \leq\]  \[\leq hp + t |S_{h, m}(A_1, \dots , A_m)|.
\]
Putting $t=(hp)^{1/2}$, we get
\[
\min(p,|A_1| + \dots + |A_m| - m + 2) - 2(hp)^{1/2} \leq\] \[ \leq \min(p,|A_1| + \dots  + |A_m| - m - (hp)^{1/2} + 2) - (hp)^{1/2}\leq \]
\[
\leq |S_{h, m}(A_1, \dots , A_m)|.
\]
\end{proof}
\begin{lemma}\label{t11} Set $\textbf{T}_{r, s}(\textbf{Z}_p)=\{A\subset \textbf{Z}_p: |A|\leq p/(r+1)s\}$. Then there exists $s$ such that
\begin{equation}\label{14} \left|\textbf{T}_{r, s}(\textbf{Z}_p)\right| \leq 2^{p/(r+1)}.\end{equation}
\end{lemma}
\begin{proof}
Let $n,m$ be positive integers, $1\leq m \leq n$. Then (see $\text{Lemma 6.8}$, \cite{SAP4})
\begin{equation}\label{iiipop}
\sum_{0\leq i\leq m}{n\choose i} \leq {\left(\frac{en}{m}\right)}^m.
\end{equation}
We choose $s$ such that
\begin{equation}\label{15}
es(r+1)\leq 2^s.
\end{equation}
Then by (\ref{iiipop}) we have (putting $n=p$ and $m=p/(r+1)s$)
\[
\left|\textbf{T}_{r, s}(\textbf{Z}_p)\right|=\sum_{0\leq i\leq p/(r+1)s}{p\choose i}\leq {\left(es(r+1)\right)}^{p/(r+1)s}\leq  {\left(2^s\right)}^{p/(r+1)s}= 2^{p/(r+1)}.
\]
\end{proof}
Let $L$ be a positive integer. For each $y  \in \{0, \dots, p-1 \}$ we define a partition $ \textbf{R}_{y, L} $ of $\textbf{Z}_p$ on the intervals of the form $J_i^y = \{(iL + 1 + y) \pmod{p}, \dots, ((i +1) L + y) \pmod{p} \}$, $0 \leq i \leq \lfloor p/L \rfloor -1$. All intervals are $J_i^y$ of $\textbf{R}_{y, L}$ have length $L$, and the set $J_y = \textbf{Z}_p \setminus \bigcup_{i}^{}{J_i^y}$ has cardinality $p-L\lfloor p/L \rfloor <L$.
The set $J_y$ is called {\it remainder} partition $\textbf{R}_{y, L} $.
In what follows we fix $y\in \{0, \dots, p-1 \}$ and consider the corresponding partition $\textbf{R}_{y, L}$. For every $A \subseteq \textbf{Z}_p $ and any integer $d$ define $d \star A = \{da \pmod {p}: a \in A \}.$ The set $d \star A $ is called {\it dilation} of $ A $. The set $A \subseteq \textbf{Z}_p$ is called {\it $L$-granular} (see \cite{GRUZ1}), if some dilation of $A$ is a union of some of the intervals $J_i^y $ (other than remainder). We denote the family of $L$-granular subsets of $\textbf{Z}_p$ by $\textbf{G}_L(\textbf{Z}_p).$
\begin{lemma}\label{t1111}  We have
\begin{equation}\label{818}
|\textbf{G}_L(\textbf{Z}_p)|  \leq  p2^{p/L}.
\end{equation}
\end{lemma}
\begin{proof} Denote the number of subsets of intervals (other than remainder) of the partition $R_{y, L}$ of $\textbf{Z}_p$ by $g(R_{y, L})$, and the number of different partitions $R_{y, L}$ of $\textbf{Z}_p$ by $r(L)$. It is obvious that
\begin{equation}\label{karab}
|\textbf{G}_L(\textbf{Z}_p)|\leq g(R_{y,L})r(L).
\end{equation}
Note that the number of intervals (other than remainder) of the partition $ \textbf{R}_{y, L}$ of $\textbf{Z}_p$ is equal to $ \lfloor p/L \rfloor$, and the number of different partitions $\textbf{R}_{y, L}$ of $\textbf{Z}_p$ is at most $p$. This and (\ref{karab}) imply the inequality (\ref{818}).
\end{proof}
\begin{lemma}\label{t13}
Let $A\subseteq \textbf{Z}_p$ have size $\alpha p,$ and let $\varepsilon_1, \varepsilon_2, \varepsilon_3$ be positive real numbers and $L>0$, $k$, $l$ be nonnegative integers satisfying $k + l \geq 2.$ Suppose that
\begin{equation}\label{3}
p > (\sqrt{8(k+l)}L)^{4^{2(k+l)}\alpha^{2(k+l-1)} \varepsilon_1^ {-2(k+l)} \varepsilon_2^{-2(k+l-1)} \varepsilon_3^{-1}}. \end{equation}
Then there exists a set $A'\subseteq \textbf{Z}_p$ with the following properties:\\
$(i)$ $A'$ is $L$-granular;\\
$(ii)$ $|A\setminus A'|\leq \varepsilon_1 p;$\\
$(iii)$ the set $kA-lA$ contains all $x\in \textbf{Z}_p$ for which \\ $(\underbrace{\chi_{A'}\ast \dots \ast\chi_{A'}}_{k}\ast \underbrace{\chi_{-A'}\ast \dots \ast\chi_{-A'}}_{l})(x) \geq {(\varepsilon_2 p)}^{k+l-1}$, with at most $\varepsilon_3 p$ exceptions.
\end{lemma}
\begin{proof}
Let $h\in \{0,\dots,p-1\},$ and $\textbf{R}_{h,L}$ be partition of $\textbf{Z}_p$.\\

$(i)$ \,\,
For given set $A \subset \textbf{Z}_p$ we define $A'\subset \textbf{Z}_p$ as the union of intervals $J_i^h $ of the partition $\textbf{R}_{h, L}$, such that $|A \cap J_i^h| \geq \varepsilon_1 L/2 $. From the definition it follows that $A'$ is $L$-granular. It is easy to see that $(-A)' = -(A')$.\\

$(ii)$ \,\,
Let $x \in A \setminus A'$. Then either $x \in J_h$ or $ x \in A \cap J_i^h$, $(i = 0, \dots, \lfloor p/L \rfloor - 1),$ and $|A\cap J_i^h| \leq \varepsilon_1 L/2.$ In the first case we have $|J_h|<L, $ and inequality (\ref{3})  implies $L \leq \varepsilon_1 p/2.$ Thus,
$$
|A\setminus A'|\leq \frac{\varepsilon_1 L}{2} \cdot \frac{p}{L} + L \leq \varepsilon_1 p.
$$

$(iii)$ \,\,
Let $\widehat{\chi_A}(x)$ be the Fourier transform of the characteristic function $\chi_A$ of $A$, so that
\[
\widehat{\chi_A}(x) = \sum_{y\in \textbf{Z}_p}{\chi_A{(y)}e^{2\pi i\frac{yx}{p}}} = \sum_{y\in A}{e^{2\pi i\frac{yx}{p}}}
\]
for all $x\in \textbf{Z}_p.$ Take $\delta=4^{-(k+l)}\varepsilon_1^{k+l} \varepsilon_2^{k+l-1} \varepsilon_3^{1/2} \alpha^{-(k+l)+ 3/2},$ where $\varepsilon_1, \varepsilon_2$ and~$\varepsilon_3$ are from inequality (\ref{3}). Set $\textbf{D} = \{x\neq 0: |\widehat{\chi_A}(x)|\geq \delta p \}.$  We define the function $f(x)$ as follows:
\[
f(x)=\frac{1}{2L-1}\sum_{j=-(L-1)}^{L-1}{e^{2\pi i\frac{jqx}{p}}}.
\]
In the future we will show that there exists $q \in \textbf{Z}_p \setminus \{0 \}$ such that for all $x \in \textbf{Z}_p$ it holds
\begin{equation}\label{4}
{|\widehat{\chi_A}(x)|}|1-{f^{k+l}(x)}|\leq \delta p.
\end{equation}
The inequality (\ref{4}) obviously holds for the case $x=0,$ since $f(0)=1$, as well as for the case $|\widehat{\chi_A}(x)| \leq \delta p,$ since $f(x) \in [-1,1].$ Thus, it remains to show the existence of $q$ such that the inequality (\ref{4}) holds for all $x \in \textbf{D}$. First we estimate the value of $1-f(x).$ Denote by $\left< x \right> $ the distance from $x$ to the nearest integer. We use the fact that $1- \cos(2\pi x) \leq 2{\pi}^2 {\left<x\right>}^2$. Then
\[
1-f(x)= \frac{2}{2L-1}\sum_{j=1}^{L-1}{\left(1- \cos\frac{2\pi jqx}{p}\right)} \leq  \frac{4{\pi}^2}{2L-1}\sum_{j=1}^{L-1}{{\left<\frac{jqx}{p}\right>}^2}\leq \]
\begin{equation}\label{5} \leq \frac{4{\pi}^2}{2L-1} {\left<\frac{qx}{p}\right>}^2\sum_{j=1}^{L-1}{j^2} \leq \frac{2{\pi}^2 L^2}{3} {\left<\frac{qx}{p}\right>}^2.
\end{equation}
Recall that for $|x|\leq 1$
\begin{equation}\label{astican}
1-x^m=(1-x)(1+x+x^2+\dots+x^{m-1})\leq m(1-x).
\end{equation}
From (\ref{5}) and (\ref{astican}) it follows
\[
{|\widehat{\chi_A}(x)|}|1-{f^{k+l}(x)}|\leq (k+l){|\widehat{\chi_A}(x)|}|1-f(x)|\leq 8(k+l)L^2 {\left<qx/p\right>}^2 {|\widehat{\chi_A}(x)|}.
\]
Note that if the inequality
\begin{equation}\label{dirixle}
\left<\frac{qx}{p}\right>\leq \frac{1}{\sqrt{8(k+l)}L}{\left(\frac{\delta p}{{|\widehat{\chi_A}(x)|}}\right)}^{1/2}
\end{equation}
holds for some $q \in \textbf{Z}_p \setminus \{0\} $ and for all $x \in \textbf{D}$ then the inequality (\ref{4}) also holds. Now we will prove that such $q$ exists. By definition, we have
$$
\left<qx/p\right>=\min\{(qx\pmod{p})/p,(p-qx\pmod{p})/p\}.
$$
Set $|\textbf{D}|=d$, $\textbf{D}=\{r_1,\dots,r_d\}$. We denote $a_i = (1/\sqrt{8(k+l)}L){\left(\delta p/{|\widehat{\chi_A}(r_i)|}\right)}^{1/2}$. Then the inequality (\ref{dirixle}) can be rewritten as
\begin{equation}\label{dirixle1}
\min\{qr_i\pmod{p}, p-qr_i\pmod{p}\}\leq pa_i, \,\,\,\,\text{where} \,\,\,\,\,\,i=1,\dots, d.
\end{equation}
Denote the set $\{(x_1,\dots,x_d) :x_1,\dots,x_d\in \textbf{Z}_p\}$ by $\textbf{Z}_p^d$. We split $\textbf{Z}_p^d$ on disjoint subsets
\[
\textbf{Z}_p^d=\bigcup_{(i_1,\dots,i_d)} {\textbf{Q}_{i_1,\dots,i_d}},
\]
where
\[
\textbf{Q}_{i_1,\dots,i_d}=\{(x_1,\dots,x_d) :i_jpa_j< x_j\leq (i_j +1)pa_j, j=1,\dots, d\}.
\]
Let $\mu_d$ be number of different sets of $\textbf{Q}_{i_1,\dots,i_d}$. Using the fact that $0\leq i_j\leq 1/a_j-1,$ $j=1,\dots,d,$ we have
$$
\mu_d\leq \prod_{i=1}^{d}{\frac{1}{a_i}}.
$$
Let us consider the following $p-1$ elements of $\textbf{Z}_p^d$:
$$
\left(qr_1\pmod{p},\dots,qr_d\pmod{p}\right),\,\,\, \text{where} \,\,\, r_1,\dots,r_d\in \textbf{D},\,\,\, q=1,\dots, p-1.
$$
We show that if
\begin{equation}\label{6}
p > \prod_{i=1}^{d}{\frac{1}{a_i}},
\end{equation}
then there exists $q$ such that for all $r_i \in \textbf{D}$, $i=1, \dots, d $, the inequality (\ref{dirixle1}) holds.
We consider two cases:\\
$(A)$ If $\mu_d = p-1$, then we take $q=q_0,$ where $q_0 \in \textbf{Z}_p \setminus \{0\} $ such that \\
$(q_0r_1 \pmod{p}, \dots, q_0r_d \pmod{p}) \in \textbf{Q}_{0, \dots, 0}$.\\
$(B)$ If $ \mu_d <p-1 $, then by pigeonhole principle, there are $q_1$, $q_2 \in \textbf{Z}_p \setminus \{0\} $ such that the vectors $(q_1r_1 \pmod{p}, \dots, q_1r_d \pmod{p})$ and $(q_2r_1 \pmod{p}, \dots, q_2r_d \pmod{p})$ belong to the same set of $\textbf{Q}_{i_1, \dots, i_d} $. Obviously, when $q = q_1-q_2 $ the inequality (\ref{dirixle1}) holds.
\par We now show that inequality (\ref{6}) is a consequence of (\ref{3}). Indeed, by the Parseval's identity, we have
\begin{equation}\label{7}
p^{-1}\left({\sum_{x\in \textbf{D}}{{|\widehat{\chi_A}(x)|}^2}  + \sum_{x\in \textbf{Z}_p\setminus \textbf{D}}{{|\widehat{\chi_A}(x)|}^2}}\right) = \sum_{x\in \textbf{Z}_p }{{|\chi_A(x)|}^2} = \alpha p.
\end{equation}
From (\ref{7}) it follows
\begin{equation}\label{8}
\sum_{x\in \textbf{D}}{{|\widehat{\chi_A}(x)|}^{2}} \leq \alpha p^2.
\end{equation}
From (\ref{8}) and the  arithmetic and geometric mean inequality, we get
\[
{\left(\prod_{x\in \textbf{D}}{{|\widehat{\chi_A}(x)|}^{2}}\right)}^{1/d}\leq \frac{1}{d}\sum_{x\in \textbf{D}}{{|\widehat{\chi_A}(x)|}^{2}}\leq \frac{\alpha p^2}{d}.
\]
i.e.
\begin{equation}\label{99000001}
\prod_{x\in \textbf{D}}{{|\widehat{\chi_A}(x)|}}\leq {\left(\frac{\alpha p^2}{d}\right)}^{d/2}.
\end{equation}
From (\ref{99000001}) we get
\begin{equation}\label{9}
{(\sqrt{8(k+l)}L)}^d {\left(\prod_{x\in \textbf{D}}{\frac{|\widehat{\chi_A}(x)|}{\delta p}}\right)}^{1/2} \leq {(\sqrt{8(k+l)}L\alpha^{1/4} \delta^{-1/2}d^{-1/4})}^d.
\end{equation}
It is easy to see that the right-hand side of (\ref{9}) is an increasing function of $d$ in the range $d < 64{(k+l)}^2 L^4 \alpha / \delta^2 e.$\\
On the other hand, from (\ref{8}) we have $d \delta^2p^2 \leq \alpha p^2.$ Hence, $d \leq\alpha / \delta^2$. Consequently
\[
{(\sqrt{8(k+l)}L\alpha^{1/4} \delta^{-1/2}d^{-1/4})}^d \leq {(\sqrt{8(k+l)}L)}^{\alpha / \delta^2}.
\]
Recall that $\delta = 4^{- (k + l)} \varepsilon_1^{k + l} \varepsilon_2^{k + l-1} \varepsilon_3^{1/2} \alpha^{- (k + l)+3/2}$. From this it follows that there exists $q$ such that the inequality (\ref{4}) holds. Moreover, without loss of generality we can assume $q = 1$ (this can be achieved by selecting an appropriate dilation of the set $A$).
\par Define two functions $ \chi_1(x)$ and $\chi_2(x)$ as follows:
\[
\chi_1(x) = \frac{1}{|\mathcal{J}|}(\chi_A \ast \chi_\mathcal{J})(x),
\]
\[
\chi_2(x) = \frac{1}{|\mathcal{J}|}(\chi_{-A} * \chi_\mathcal{J})(x),
\]
where $\mathcal{J} = \{-(L-1), \dots , L-1\}.$ From (\ref{7892}) it follows that
\begin{equation}\label{wwwww}
\chi_1(x) = \frac{1}{|\mathcal{J}|}|A\cap (\mathcal{J} + x)|,
\end{equation}
\begin{equation}\label{wwwww0999}
\chi_2(x) = \frac{1}{|\mathcal{J}|}|(-A)\cap (\mathcal{J} + x)|,
\end{equation}
and from (\ref{01987199}) we have $\widehat{\chi_1}(x) = \widehat{\chi_A}(x)f(x)$ and $\widehat{\chi_2}(x) = \widehat{\chi_{-A}}(x)f(x)$. Hence, by Parseval's  identity and from (\ref{01987199}) we get
\[
\sum_{x\in \textbf{Z}_p}{{\left|(\underbrace{\chi_A \ast \dots \ast \chi_A}_{k}\ast \underbrace{\chi_{-A} \ast \dots \ast \chi_{-A}}_{l})(x)-(\underbrace{\chi_1 \ast \dots \ast \chi_1}_{k}\ast \underbrace{\chi_2 \ast \dots \ast \chi_2}_{l})(x)\right|}^2} =
\]
\[
=\!\!p^{-1}\!\!\!\sum_{x\in \textbf{Z}_p}{{\!\!\left|(\widehat{\underbrace{\chi_A \ast \dots \ast \chi_{A}}_{k}\ast \underbrace{\chi_{-A} \ast \dots \ast \chi_{-A}}_{l}})(x)\! -\!(\widehat{\underbrace{\chi_1 \ast \dots \ast \chi_{1}}_{k}\ast \underbrace{\chi_2 \ast \dots \ast \chi_2}_{l}})(x)\right|}^2}
\]
\[
=p^{-1}\sum_{x\in \textbf{Z}_p}{{\left|{\widehat{\chi_A}^{k}(x)}{\widehat{\chi_{-A}}^{l}(x)} -{\widehat{\chi_1}^{k}(x)}{\widehat{\chi_2}^{l}(x)}\right|}^2}=
\]
\[
=p^{-1}\sum_{x\in \textbf{Z}_p}{{{\left|\widehat{\chi_A}(x)\right|}^{2k}}{{\left|\widehat{\chi_{-A}}(x)\right|}^{2l}}{\left|1 -{f^{k+l}(x)}\right|}^2}\leq
\]
\begin{equation}\label{10}
\leq p^{-1}{\left(\sup_{x\in \textbf{Z}_p}{{{{\left|\widehat{\chi_A}(x)\right|}^{k-1}}}{{{\left|\widehat{\chi_{-A}}(x)\right|}^{l}}} \left|1-{f^{k+l}(x)}\right|}\right)}^2\sum_{x\in \textbf{Z}_p}{{\left|\widehat{\chi_A}(x)\right|}^{2}}.
\end{equation}
We have
\begin{equation}\label{tran}
\left|\widehat{\chi_A}(x)\right|=\left|\sum_{y\in \textbf{Z}_p}{\chi_A{(y)}e^{2\pi i\frac{yx}{p}}}\right| = \left|\sum_{y\in A}{e^{2\pi i\frac{yx}{p}}}\right|\leq \sum_{y\in A}{\left|e^{2\pi i\frac{yx}{p}}\right|}= \alpha p,
\end{equation}
\begin{equation}\label{trancho}
\left|\widehat{\chi_{-A}}(x)\right|=\left|\sum_{y\in \textbf{Z}_p}{\chi_{-A}{(y)}e^{2\pi i\frac{yx}{p}}}\right| = \left|\sum_{y\in -A}{e^{2\pi i\frac{yx}{p}}}\right|\leq \sum_{y\in -A}{\left|e^{2\pi i\frac{yx}{p}}\right|}= \alpha p.
\end{equation}
From (\ref{4}), (\ref{7}), (\ref{10}), (\ref{tran}) and (\ref{trancho}) it follows
\[
\sum_{x\in \textbf{Z}_p}{{\left|(\underbrace{\chi_A \ast \dots \ast \chi_A}_{k}\ast \underbrace{\chi_{-A} \ast \dots \ast \chi_{-A}}_{l})(x)-(\underbrace{\chi_1 \ast \dots \ast \chi_1}_{k}\ast \underbrace{\chi_2 \ast \dots \ast \chi_2}_{l})(x)\right|}^2} \leq
\]
\[
\leq {\left(\sup_{x\in \textbf{Z}_p}{{|\widehat{\chi_A}(x)|}\left|1-{f^{k+l}(x)}\right|}\right)}^2 {\alpha}^{2(k+l)-3} p^{2(k+l)-3} \leq
\]
\begin{equation}\label{11}
\leq {\alpha}^{2(k+l)-3} \delta^2 p^{2(k+l)-1}.
\end{equation}
Suppose that $x \in A'$ ($x \in -A'$). Then there exists an interval $\mathcal{I}$ of length $L$ such that $\mathcal{I}\subseteq \{x-(L-1), \dots, x + (L-1)\}$ and $x\in \mathcal{I}.$ From definition of $A'$ ($-A'$) it follows that $|\mathcal{I} \cap A|\geq \varepsilon_1 L/2$ ($|\mathcal{I} \cap (-A)|\geq \varepsilon_1 L/2$).
From the definition of $\chi_1(x)$ ($\chi_2(x)$) it follows that $\chi_1(x) \geq \varepsilon_1 /4$ ($\chi_2(x)\geq \varepsilon_1 /4$). Observe, that $\chi_1(x) \geq \varepsilon_1 \chi_{A'}{(x)}/4$ and $\chi_2(x) \geq \varepsilon_1 \chi_{-A'}{(x)}/4$ hold for all $x \in \textbf{Z}_p$ . From this and (\ref{7892}) it follows that
\[
(\underbrace{\chi_1 \ast \dots \ast \chi_1}_{k}\ast \underbrace{\chi_2 \ast \dots \ast \chi_2}_{l})(x) \geq
\]
\begin{equation}\label{klro1}
\geq \varepsilon_1^{k+l} (\underbrace{\chi_{A'} \ast \dots \ast \chi_{A'}}_{k}\ast \underbrace{\chi_{-A'} \ast \dots \ast \chi_{-A'}}_{l})(x)/ 4^{k+l}
\end{equation}
for all $x \in \textbf{Z}_p.$ In the case
\begin{equation}\label{12}(\underbrace{\chi_{A'} \ast \dots \ast \chi_{A'}}_{k}\ast \underbrace{\chi_{-A'} \ast \dots \ast \chi_{-A'}}_{l})(x) \geq {(\varepsilon_2  p)}^{k+l-1},
\end{equation}
by (\ref{klro1}) we have
\begin{equation}\label{klro}
(\underbrace{\chi_1 \ast \dots \ast \chi_1}_{k} \ast \underbrace{\chi_2 \ast \dots \ast \chi_2}_{l})(x) \geq {\varepsilon_1}^{k+l} {(\varepsilon_2  p)}^{k+l-1}/4^{k+l}.
\end{equation}
Now we show that the number of elements $x \in \textbf{Z}_p$ such that satisfying (\ref{12}) and $(\underbrace{\chi_{A} \ast \dots \ast \chi_{A}}_{k} \ast \underbrace{\chi_{-A} \ast \dots \ast \chi_{-A}}_{l})(x) =0$, does not exceed $\varepsilon_3p$. Denote the set of such elements by $\textbf{F}$. Observe, that for every $x\in \textbf{F}$
\[
{|(\underbrace{\chi_A \ast \dots \ast \chi_A}_{k}\ast \underbrace{\chi_{-A} \ast \dots \ast \chi_{-A}}_{l})(x)-(\underbrace{\chi_1 \ast \dots \ast \chi_1}_{k}\ast \underbrace{\chi_2 \ast \dots \ast \chi_2}_{l})(x)|}^2 \geq
\]
\begin{equation}\label{13}
\geq \frac{{\varepsilon_1}^{2(k+l)} {\varepsilon_2}^{2(k+l-1)} p^{2(k+l-1)} } {4^{2(k+l)}}.
\end{equation}
By (\ref{11}) and~(\ref{13})
\[
{\alpha}^{2(k+l)-3} \delta^2 p^{2(k+l)-1}\geq
\]
\[
\geq \sum_{x\in \textbf{Z}_p}{{\left|(\underbrace{\chi_A \ast \dots \ast \chi_A}_{k}\ast \underbrace{\chi_{-A} \ast \dots \ast \chi_{-A}}_{l})(x)-(\underbrace{\chi_1 \ast \dots \ast \chi_1}_{k}\ast \underbrace{\chi_2 \ast \dots \ast \chi_2}_{l})(x)\right|}^2}=
\]
\[
=\!\!\sum_{x\in \textbf{F}}{{\!\left|(\underbrace{\chi_A \ast \dots \ast \chi_A}_{k}\ast \underbrace{\chi_{-A} \ast \dots \ast \chi_{-A}}_{l})(x)\!-\!(\underbrace{\chi_1 \ast \dots \ast \chi_1}_{k}\ast \underbrace{\chi_2 \ast \dots \ast \chi_2}_{l})(x)\right|}^2}+
\]
\[
+ \!\!\sum_{x\in (\textbf{Z}_p\setminus \textbf{F})}{{\!\left|(\underbrace{\chi_A \ast \dots \ast \chi_A}_{k}\ast \underbrace{\chi_{-A} \ast \dots \ast \chi_{-A}}_{l})(x)\!-\!(\underbrace{\chi_1 \ast \dots \ast \chi_1}_{k}\ast \underbrace{\chi_2 \ast \dots \ast \chi_2}_{l})(x)\right|}^2}
\]
\[
\geq |\textbf{F}|\frac{{\varepsilon_1}^{2(k+l)} {\varepsilon_2}^{2(k+l-1)} p^{2(k+l-1)} } {4^{2(k+l)}}+
\]
\[
+\!\! \sum_{x\in (\textbf{Z}_p\setminus \textbf{F})}{{\!\left|(\underbrace{\chi_A \ast \dots \ast \chi_A}_{k}\ast \underbrace{\chi_{-A} \ast \dots \ast \chi_{-A}}_{l})(x)\!-\!(\underbrace{\chi_1 \ast \dots \ast \chi_1}_{k}\ast \underbrace{\chi_2 \ast \dots \ast \chi_2}_{l})(x)\right|}^2}.
\]
This implies
\[
|\textbf{F}| \leq \frac{4^{2(k+l)} \alpha^{2(k+l)-3} \delta^2}{{\varepsilon_1}^{2(k+l)} {\varepsilon_2}^{2(k+l-1)}} p\leq \varepsilon_3 p.
\]
\end{proof}

\section{The proof of Theorem \ref{t6}}
\subsection{The upper bound}
Let $k$,$l$ be nonnegative integers with $k+l\geq 2.$ Suppose that $s$ satisfies $es(k + l +1)\leq 2^s$. We divide a partition of $\textbf{SS}_{k, l}(\textbf{Z}_p)$ into two parts:
\begin{equation}\label{unnnnn}
\textbf{SS}_{k, l}(\textbf{Z}_p)=\textbf{SS}'_{k, l, s}(\textbf{Z}_p)\cup \textbf{SS}''_{k, l, s}(\textbf{Z}_p),
\end{equation}
where
$$\textbf{SS}'_{k, l, s}(\textbf{Z}_p)=\{B\in \textbf{SS}_{k, l}(\textbf{Z}_p) : B=kA-lA \,\,\, \text{and}\,\,\, |A|\leq p/(k+l+1)s\},$$
$$\textbf{SS}''_{k, l, s}(\textbf{Z}_p)=\{B\in \textbf{SS}_{k, l}(\textbf{Z}_p) : B=kA-lA \,\,\, \text{and}\,\,\, |A|> p/(k+l+1)s\}.$$
It is obvious that
\begin{equation}\label{unnnnn1}
|\textbf{SS}_{k, l}(\textbf{Z}_p)|\leq |\textbf{SS}'_{k, l, s}(\textbf{Z}_p)|+ |\textbf{SS}''_{k, l, s}(\textbf{Z}_p)|.
\end{equation}
Since every set $A \subseteq \textbf{Z}_p$ generates one set of the form $kA-lA$ we obtain
\begin{equation}\label{trem}
\left|\textbf{SS}'_{k, l, s}(\textbf{Z}_p)\right|\leq \left|\textbf{T}_{k+l, s}(\textbf{Z}_p)\right|.
\end{equation}
By (\ref{t11}) and (\ref{trem})  we have
\begin{equation}\label{trem1}
\left|\textbf{SS}'_{k, l, s}(\textbf{Z}_p)\right|\leq 2^{p/(k+l+1)}.
\end{equation}
\par Now we prove an upper bound for $|\textbf{SS}''_{k, l, s}(\textbf{Z}_p)|$. Suppose that the cardinality of $A \subseteq \textbf{Z}_p$ is larger than $p/(k+l+1)s.$ Let $p$ be a prime number such that for some nonnegative integers $k, l, L>0$ and positive real numbers $\varepsilon_1, \varepsilon_2$ and $\varepsilon_3$ the condition (\ref{3}) is fulfilled. By Lemma \ref{t13} there exists a subset $A'$ with properties $(i)-(iii).$ We estimate the number of $(k,l)$-sumsets $kA-lA$ by counting pairs $(A', kA-lA).$
\par Now let $A'\in \textbf{G}_L(\textbf{Z}_p)$  be given. For any subset $C \subseteq \textbf{Z}_p$ we denote by $\overline{C}$ the complement of the subset $C$ in $\textbf{Z}_p$.
\par If $|A'| \geq p/(k + l +1)$, then from $(iii)$ of Lemma \ref{t13} we obtain that
$\overline{kA-lA}$ is a subset of the union of the  set
$\overline{S_{{(\varepsilon_2 p)}^{k+l-1}, k+l}(\underbrace{\chi_{A'}, \dots ,\chi_{A'}}_{k},\underbrace{ \chi_{-A'}, \dots ,\chi_{-A'}}_{l})}$ and a set of cardinality not exceeding $\varepsilon_3 p$. By Lemma \ref{t10} we have
\[
|S_{{(\varepsilon_2 p)}^{k+l-1}, k+l}(\underbrace{\chi_{A'}, \dots ,\chi_{A'}}_{k},\underbrace{ \chi_{-A'}, \dots ,\chi_{-A'}}_{l})|\geq
\]
\[
\geq \min(p,(k+l)|A'| - (k+l) + 2) - 2({(\varepsilon_2 p)}^{k+l-1}p)^{1/2}.
\]
If $|A'|\geq p/(k + l +1)$, we obtain
\[
|\overline{S_{{(\varepsilon_2 p)}^{k+l-1}, k+l}(\underbrace{\chi_{A'}, \dots ,\chi_{A'}}_{k},\underbrace{ \chi_{-A'}, \dots ,\chi_{-A'}}_{l})}|=
\]
\[
=p- |S_{{(\varepsilon_2 p)}^{k+l-1}, k+l}(\underbrace{\chi_{A'}, \dots ,\chi_{A'}}_{k},\underbrace{ \chi_{-A'}, \dots ,\chi_{-A'}}_{l})| \leq
\]
\[
\leq p/(k+l+1)+2{\varepsilon_2}^{(k+l-1)/2}p^{(k+l)/2} + (k+l - 2).
\]
\par It is obvious that for any subset $B \subseteq \textbf{Z}_p$ the set $kB-lB$ uniquely determines the set $\overline{kB-lB}$. From above it follows that the number of choices $kA-lA$ for given $A'$ of cardinality exceeding $p/(k+l+1) $, is at most
\begin{equation}\label{bozer1}
 2^{p/(k+l+1) +(k+l -2)+({2\varepsilon_2}^{(k+l-1)/2} p^{(k+l-2)/2} + \varepsilon_3)p}.
\end{equation}
\par If $|A'|<p/(k+l+1)$, then by $(i)$ of Lemma \ref{t13} we have $|A \setminus A'| \leq \varepsilon_1p$ . This implies  that $|A| \leq |A'| + \varepsilon_1 p$. Since every set $A \subseteq \textbf{Z}_p$ generates exactly one set of form $kA-lA$, we obtain that the number of choices $kA-lA$ for given $A'$ of cardinality not exceeding $p/(k+l+1) $, is at most
\begin{equation}\label{bozer2}
 2^{p/(k+l+1)+\varepsilon_1 p}.
\end{equation}
From (\ref{bozer1}),\,\,\,(\ref{bozer2}), Lemma \ref{t1111} by applying Lemma \ref{t13} with parameters $\varepsilon_1 = \varepsilon_3 = \varepsilon,$ $L = 1 + \lfloor1/\varepsilon \rfloor$ and~$\varepsilon_2 = \varepsilon ^{2/(k+l-1)} p^{(2-k-l)/(k+l-1)}$, we obtain
\begin{equation}\label{opoas}
|\textbf{SS}''_{k, l, s}(\textbf{Z}_p)| \leq 2^{(p/(k+l+1)) + (k+l-2) + o(p)}.
\end{equation}
From (\ref{unnnnn1}),\,\,\,(\ref{trem1}) and (\ref{opoas}) it follows that
\[
|\textbf{SS}_{k, l}(\textbf{Z}_p)| \leq 2^{p/(k+l+1)} + 2^{(p/(k+l+1)) + (k+l-2) + o(p)}=2^{(p/(k+l+1))+(k+l-2)+o(p)}.
\]

\subsection{The lower bound}
Set  $\textbf{SS}_{k, l}(\textbf{Z}_p, \mathbb{P}) =\{A: \mathbb{P}\subseteq A,\, A\in \textbf{SS}_{k, l}(\textbf{Z}_p)\}$ and $L=\lfloor p/(2(k+l)-1)\rfloor -1$.
\begin{lemma}\label{30}
Let $k$,$l$ be nonnegative integers with $k+l\geq 2,$ and let $\mathbb{P}\subseteq \textbf{Z}_p$ be arbitrary arithmetic progression of length $(k+l)(L-1)+1$. Then there exists a positive constant $\textbf{C}_{k, l}$ such that
\[
|\textbf{SS}_{k, l}(\textbf{Z}_p, \mathbb{P})| \geq \textbf{C}_{k, l}2^{p/(2(k+l)-1)}.
\]
\end{lemma}
\begin{proof}
Without loss of generality we assume $\mathbb{P}=\{k-lL, \dots , kL-l\}$. All of our sets will be of the form 
$$
A=A(B)=k(B\cup \{-(2L +1), 2L+1\})-l(B\cup \{-(2L +1), 2L+1\}),
$$ 
where $B\subseteq \{-L, -L+1, \dots, L \}$ and $-B=B$. It is easy to see that different sets $B\subseteq \{-L, -L+1, \dots, L \}$ generate different sets $A(B)$.
\par Set $N_{k,l}=\lceil\log{}{(8(k+l)^2)}/\log{}{(4/3)}\rceil$ and
$$
X=\{0,1,\dots, N_{k,l}\}\cup \bigcup_{i=1}^{k+l-1}{\left(\lfloor(i+1)L/(k+l)\rfloor -N_{k,l},\dots,\lceil(i+1) L/(k+l)\rceil \right)}.
$$
We define the set $B\subseteq \{-L, -L+1, \dots, L \}$ as follows:
$$B=B(C)=-C \cup C\cup X \cup -X,$$
where elements of the set $C$ are picked from the set $\{1, \dots , L\}\setminus X$ randomly, independently, with probability $1/2$. Set
$$Y=\{0\}\cup \{k+l,\dots, (k+l)N_{k,l}\}\cup \bigcup_{i=1}^{k+l-1}{\{(i+1)L-(k+l)N_{k,l},\dots,(i+1)L\}}.$$
It is obvious that $-Y\cup Y\subseteq kB-lB.$ If $x\notin kB-lB$, then in the representation $x$ in the form $x=x_1+\dots+x_k-x_{k+1}-\dots-x_{k+l},$ there exists at least one $x_i$\,\,\, $(i\in \{1,\dots,k+l\})$ such that $x_i\notin B$. Set
$$
\mathcal{Q}(x)= \{(x_1,\dots,x_{k+l}): x=\sum_{i=1}^{k}{x_i}-\sum_{j=k+1}^{k+l}{x_j}, 
x_1,\dots, x_{k+l}\in \{-L,\dots, L\}\},
$$
and suppose that $|\mathcal{Q}(x)|=q$.
\par We say that the vectors $(x_1,\dots, x_{k+l})$ and $(y_1,\dots, y_{k+l})$ do not intersect, if $\{x_1,\dots, x_{k+l}\}\cap \{y_1,\dots, y_{k+l}\}=\emptyset.$
\par Set $\mathcal{R}_0=\{(k+l)N_{k,l}+1, \dots , L\}.$ We show that for every $x\in \mathcal{-R}_0\cup \mathcal{R}_0$ the following inequality
\begin{equation}\label{puchik33}
\textbf{Pr}(x\notin kB-lB)\leq {\left(\frac{3}{4}\right)}^{\left\lfloor \frac{|x|}{k+l} \right\rfloor}
\end{equation}
holds. We have
\[
\textbf{Pr}(x\notin kB-lB)=
\]
\[
=\textbf{Pr}
((x_1^1+\dots+x_k^1-x_{k+1}^1-\dots-x_{k+l}^1\notin kB-lB) \& \dots
\]
\[
\dots \&(x_1^q+\dots+x_k^q-x_{k+1}^q-\dots-x_{k+l}^q\notin kB-lB)
)\leq
\]
\[
\leq \textbf{Pr}
((x_1^{11}+\dots+x_k^{11}-x_{k+1}^{11}-\dots-x_{k+l}^{11}\notin kB-lB) \& \dots
\]
\[
\dots \&(x_1^{1n}+\dots+x_k^{1n}-x_{k+1}^{1n}-\dots-x_{k+l}^{1n}\notin kB-lB)
)=
\]
\[
=\textbf{Pr}\left((x_1^{11}\notin B \vee \dots \vee x_{k+l}^{11}\notin B) \& \dots \&(x_1^{1n} \notin B\vee  \dots \vee x_{k+l}^{1n}\notin B)\right)=
\]
\[
=\textbf{Pr}\left((x_1^{11}\notin B) \vee \dots \vee (x_{k+l}^{11}\notin B)\right)\cdot
... \cdot \textbf{Pr}\left((x_1^{1n} \notin B)\vee  \dots \vee (x_{k+l}^{1n}\notin B)\right)=
\]
\[
=\textbf{Pr}\left(\overline{(x_1^{11}\in B) \& \dots \& (x_{k+l}^{11}\in B)}\right)\times \dots
\]
\[
\dots \times \textbf{Pr}\left(\overline{(x_1^{1n}\in B) \& \dots \& (x_{k+l}^{1n}\in B)}\right)=
\]
\[
=\left(1-\textbf{Pr}\left((x_1^{11}\in B) \& \dots \& (x_{k+l}^{11}\in B)\right)\right)\times \dots
\]
\begin{equation}\label{morgan2}
\dots \times \left(1-\textbf{Pr}\left((x_1^{1n}\in B) \& \dots \& (x_{k+l}^{1n}\in B)\right)\right),
\end{equation}
where the vectors $(x_1^i,\dots,x_{k+l}^i)\in \mathcal{Q}(x), i=1,\dots, q$, and the vectors $(x_1^{1j},\dots,$ $ x_{k+l}^{1j}),$ $j=1,\dots ,n\leq q,$ are pairwise disjoint. 
\par Note that the vectors $(x-i(k+l-1),\underbrace{i, \dots , i}_{k-1}, \underbrace{-i, \dots , -i}_{l})$ are pairwise disjoint for every $x\in \mathcal{-R}_0,$ where $-\left\lfloor |x|/(k+l)\right\rfloor\leq i\leq -1,$ and  $x\in \mathcal{R}_0,$ where $1 \leq i\leq \left\lfloor x/(k+l)\right\rfloor.$ From this and (\ref{morgan2}) we obtain the inequality (\ref{puchik33}).
\par Set $\mathcal{L}_j= \{jL+1, \dots , (j+1)L-(k+l)N_{k,l}-1\},$ $j=1,\dots, k+l-1.$ Similarly to the inequality (\ref{puchik33}) we have
\begin{equation}\label{puchikklirik33}
\textbf{Pr}(x\notin kB-lB)\leq {\left(\frac{3}{4}\right)}^{\left\lfloor \frac{(j+1)L -|x|}{k+l} \right\rfloor},
\end{equation}
where $x\in \mathcal{-L}_j\cup \mathcal{L}_j,$ $j=1,\dots, k+l-1.$
\par From (\ref{puchik33}) and (\ref{puchikklirik33}) it is easy to see that

\vspace*{-12pt}
\begin{equation}\label{hav}
\textbf{Pr}\left(\mathbb{P}\nsubseteq kB-lB\right)\leq (k+l)\sum_{x\geq (k+l)N_{k,l}+1}{{\left(\frac{3}{4}\right)}^{\left\lfloor \frac{x}{k+l} \right\rfloor}}.
\end{equation}
Note that if $N_{k,l}\geq  \log{}{(8(k+l)^2)}/\log{}{(4/3)},$ the right-hand side of (\ref{hav}) does not exceed $1/2$. This leads that there exists at least $2^{L-(k+l)N_{k,l}-1}$ subsets $B\subseteq \{-L,-L+1, \dots ,L\}$ such that $\mathbb{P}\subseteq kB-lB$.
\end{proof}
\par Let $k$,$l$ be nonnegative integers with $k+l\geq 2,$ and let $\mathbb{P}\subseteq \textbf{Z}_p$ be arbitrary arithmetic progression of length $(k+l)(L-1)+1$. By Lemma \ref{30} we have
\[
|SS_{k,l}(\textbf{Z}_p)|\geq |\textbf{SS}_{k,l}(\textbf{Z}_p, \mathbb{P})|\geq \textbf{C}_{k,l}2^{p/(2(k+l)-1)}.
\]

\vspace*{-12pt}
\section*{Acknowledgements}

\vspace*{-6pt}
The author expresses his gratitude to his supervisor Prof. A. A. Sapozhenko for the statement of the problem and attention to this work. The research was supported by RFBR project Nr 10-01-00768-a.

\bigskip
\rightline{\emph{Received: March 8, 2012 {\tiny \raisebox{2pt}{$\bullet$\!}} Revised: April 16, 2012}} 

\end{document}